\documentclass[a4paper, 10pt, twosides]{amsart}

\usepackage{lmodern}
\usepackage{amssymb}
\usepackage{amsthm}
\usepackage{amsaddr}
\usepackage{mathtools}
\usepackage{MnSymbol}
\usepackage{graphicx}
\usepackage{caption}
\usepackage{subcaption}
\usepackage{url}

\usepackage{enumitem}
\setitemize{nolistsep}
\setenumerate{nolistsep}
\usepackage{todonotes}
\usepackage{algorithm,algorithmic}

\allowdisplaybreaks

\newtheorem{thm}{Theorem}

\theoremstyle{definition}
\newtheorem{defn}{Definition}

\theoremstyle{remark}
\newtheorem{rem}{Remark}

\theoremstyle{assumption}
\newtheorem{assump}{Assumption}

\theoremstyle{fact}
\newtheorem{fact}{Fact}

\theoremstyle{claim}

\theoremstyle{prob}
\newtheorem{prob}{Problem}

\theoremstyle{algo}

\theoremstyle{experiment}

\numberwithin{equation}{section}
\newcommand{\abs}[1]{\left\lvert{#1}\right\rvert}

\newcommand{\norm}[1]{\left\lVert#1\right\rVert}

\newcommand{\R}{\mathbb{R}}
\newcommand{\N}{\mathbb{N}}

\renewcommand{\P}{\mathcal{P}}

\title[]{A note on the existence of stabilizing switching signals for switched linear systems}
\author{Atreyee Kundu}
\address{Department of Electrical Engineering,\\Indian Institute of Science Bangalore,\\Bengaluru - 560012, India,\\ E-mail: atreyeek@iisc.ac.in}


\date{\today}

\begin{document}

	\begin{abstract}
       This paper deals with stability of discrete-time switched linear systems whose all subsystems are unstable. We present sufficient conditions on the subsystems matrices such that a switched system is globally exponentially stable under a set of purely time-dependent switching signals that are allowed to activate all subsystems. The main apparatuses for our analysis are (matrix) commutation relations between certain products of the subsystems matrices and graph-theoretic arguments. We present a numerical experiment to demonstrate our results.
    \end{abstract}

    \maketitle
\section{Introduction}
\label{s:intro}
\subsection{The problem}
\label{ss:prob_stat}
	A \emph{switched system} has two ingredients --- a family of systems and a switching signal. The \emph{switching signal} selects an \emph{active subsystem} at every instant of time, i.e., the system from the family whose dynamics is currently being followed \cite[\S 1.1.2]{Liberzon2003}. Switched systems find wide applications in power systems and power electronics, automotive control, aircraft and air traffic control, network and congestion control, etc. \cite[p.\ 5]{Sun}.
    
    We consider a discrete-time switched linear system \cite[\S 1.1.2]{Liberzon2003}
    \begin{align}
    \label{e:swsys}
        x(t+1) = A_{\sigma(t)}x(t),\:\:x(0)=x_{0},\:\:t\in\N_{0}
    \end{align}
    generated by
    \begin{itemize}[label = \(\circ\), leftmargin = *]
        \item a family of systems
        \begin{align}
        \label{e:family}
            x(t+1) = A_{\ell}x(t),\:\:x(0) = x_{0},\:\:\ell\in\P,\:\:t\in\N_{0},
        \end{align}
        where \(x(t)\in\R^{d}\) is the vector of states at time \(t\), \(\P = \{1,2,\ldots,N\}\) is an index set, \(A_{\ell}\in\R^{d\times d}\), \(\ell\in\P\) are constant matrices, and
        \item a switching signal \(\sigma:\N_{0}\to\P\) that specifies at every time \(t\), the index of the {active subsystem}, \(i\in\P\).
    \end{itemize}
    The solution to \eqref{e:swsys} is given by
    \begin{align}
    \label{e:soln}
        x(t) = A_{\sigma(t-1)}\ldots A_{\sigma(1)}A_{\sigma(0)}x_{0},\:\:t\in\N,
    \end{align}
    where we have suppressed the dependence of \(x\) on \(\sigma\) for notational simplicity.

    It is well-known that the switched system \eqref{e:swsys} does not necessarily inherit qualitative properties of its constituent subsystems in \eqref{e:family}. For instance, divergent trajectories may be generated by switching appropriately among stable \(A_{\ell}\), \(\ell\in\P\), while a suitably constrained \(\sigma\) may ensure stability of \eqref{e:swsys} even if all \(A_{\ell}\), \(\ell\in\P\), are unstable.\footnote{A matrix \(A\in\R^{d\times d}\) is Schur stable if all its eigenvalues are inside the open unit disk. We call \(A\) unstable if it is not Schur stable.} We will operate under the following assumption:
    \begin{assump}
    \label{a:all_unstable}
    \rm{
        All \(A_{\ell}\), \(\ell\in\P\), are unstable.
    }
    \end{assump}

    Given a family of systems \eqref{e:family} such that Assumption \ref{a:all_unstable} holds, the existence of a switching signal \(\sigma\) under which the switched system \eqref{e:swsys} is stable, clearly depends on the properties of the subsystem matrices \(A_{\ell}\), \(\ell\in\P\). Our objective is to find conditions on \(A_{\ell}\), \(\ell\in\P\), such that they admit switching signals under which the switched system \eqref{e:swsys} is stable. In other words, we are interested in stabilizability \cite{Fiacchini2016} of \eqref{e:swsys}. Recall that
    \begin{defn}{\cite[\S2]{Agrachev2012}}
    \label{d:gues}
    \rm{
        The switched system \eqref{e:swsys} is \emph{globally exponentially stable (GES) under a switching signal \(\sigma\)} if there exist positive numbers \(c\) and \(\lambda\) such that for arbitrary choices of the initial condition \(x_{0}\), the following inequality holds:
        \begin{align}
        \label{e:gues1}
            \norm{x(t)}\leq ce^{-\lambda t}\norm{x_{0}}\:\:\text{for all}\:t\in\N,
        \end{align}
        where \(\norm{v}\) denotes the Euclidean norm of a vector \(v\).
    }
    \end{defn}
    \begin{defn}
    \label{d:stabilizability}
    \rm{
        The switched system \eqref{e:swsys} is called \emph{stabilizable} if there exists a switching signal \(\sigma\) under which \eqref{e:swsys} is
        GES.
    }
    \end{defn}
    We will solve the following problem:
    \begin{prob}
    \label{prob:mainprob}
    \rm{
        Find conditions on the matrices \(A_{\ell}\), \(\ell\in\P\), such that the switched system \eqref{e:swsys} is stabilizable.
    }
    \end{prob}

\subsection{Literature survey}
\label{ss:lit_survey}
    Problem \ref{prob:mainprob} has attracted considerable research attention in the past few decades. Given a set of unstable matrices \(A_{\ell}\), \(\ell\in\P\), the problem of deciding whether (or not) there exists a \(\sigma\) that stabilizes \eqref{e:swsys}, in general, belongs to the class of NP-hard problems, see \cite{Skafidas1999,Vlassis2014} for results and discussions.

    On the one hand, necessary and sufficient conditions for this problem is proposed only recently in \cite{Fiacchini2014}. Verifying the condition of \cite{Fiacchini2014} involves checking the containment of a set in the union of other sets, and hence possesses inherent computational complexity. On the other hand, sufficient conditions for determining the existence of a stabilizing \(\sigma\) are plenty in the literature. Stability of \eqref{e:swsys} in the setting of the so-called min-switching signals \cite{Liberzon2003} is studied in \cite{Geromel2006}. The matrices \(A_{\ell}\), \(\ell\in\P\), are required to satisfy a set of bilinear matrix inequalities (BMIs) called the Lyapunov-Metzler inequalities, for stability under these switching signals. In \cite{Sun2011} the existence of min-switching signals is claimed to be necessary and sufficient for exponential stabilizability of a switched system. Algebraic relations of Lyapunov-Metzler inequalities with classical S-procedure characterization \cite{Lure1944} are explored recently in \cite{Heemels2017}. Generalized versions of both sets of inequalities along with their relations to the classical ones are studied in \cite{Kundu2017, Fiacchini2016}. In \cite{Fiacchini2016} stability of \eqref{e:swsys} is addressed under periodic switching signals based on satisfaction of a set of linear matrix inequalities (LMIs). The proposed LMI conditions are equivalent to the generalized versions of Lyapunov-Metzler inequalities, and is implied by the classical Lyapunov-Metzler inequalities \cite{Fiacchini2016}.
    
    Stabilizability of switched systems under Assumption \ref{a:all_unstable} and various types of pre-specified restrictions on switching signals has also been studied earlier in the literature. In \cite{Fiacchini2018} the authors address stabilization of a switched linear system under a set of switching signals restricted by the language of a non-deterministic finite automaton \cite{Hopcroft} . Such an automaton captures a large class of constraints on the switching signals. The authors present an algorithm to design stabilizing state-dependent switching signals. It is shown using geometry of certain sets that the termination of this algorithm is a necessary and sufficient condition for recurrent stabilizability, which in turn is a sufficient condition for stabilizability of a switched system. In \cite{Kun_arxiv2020} the author studies stability of switched linear systems whose switching signals obey pre-specified restrictions on admissible switches between the subsystems and admissible dwell times on the subsystems. Sufficient conditions involving (matrix) commutation relations between certain products of the subsystems matrices and graph-theoretic arguments are derived.

\subsection{Our contributions}
\label{ss:contri}
   In this paper we propose a new sufficient condition for stabilizability of the switched system \eqref{e:swsys}. Our main contributions are twofold:
    \begin{itemize}[label = \(\circ\), leftmargin = *]
        \item First, we deal with purely time-dependent switching signals that are not restricted to periodic constructions, and
        \item second, we transcend beyond the regime of matrix inequalities based stabilizability conditions, and rely on properties of (matrix) commutators between certain combinations of the subsystem matrices \(A_{\ell}\), \(\ell\in\P\), for this purpose.
    \end{itemize}

    We operate under the assumption that \eqref{e:family} admits two subsystems \(i,j\in\P\), that form a Schur stable combination \(A_{i}^{p}A_{j}^{q}\) for some \(p,q\in\N\). This implies the existence of a stabilizing periodic switching signal that dwells on subsystem \(j\) for \(q\) units of time and then dwells on subsystem \(i\) for \(p\) units of time, and this procedure is repeated. However, we aim for the existence of more general switching signals that are not necessarily periodic.

    Our switching signals are characterized in terms of admissible switches between the subsystems and admissible dwell times on the subsystems. We employ infinite walks on a directed graph for this purpose. Towards ensuring stability of the switched system \eqref{e:swsys} under our switching signals, we rely on the rate of decay of the Schur stable combination formed by two unstable subsystem matrices \(A_{i}\) and \(A_{j}\), \(i,j\in\P\), upper bounds on the norms of the commutators of the subsystem matrices \(A_{\ell}\), \(\ell\in\P\), with the matrix \(A_{i}^{p}A_{j}^{q}\), and certain scalars capturing the properties of these matrices.
    In particular, we solve Problem \ref{prob:mainprob} in two steps:
    \begin{itemize}[label = \(\circ\), leftmargin = *]
    	\item First, given a family of systems \eqref{e:family} such that Assumptions \ref{a:all_unstable} holds, we propose an algorithm to construct purely time-dependent switching signals that are not restricted to periodic construction, and
	\item second, we identify conditions on the subsystem matrices \(A_{\ell}\), \(\ell\in\P\), such that the switched system \eqref{e:swsys} is GES under our switching signals.
    \end{itemize}

    Matrix commutators (Lie brackets) have been employed to study stability of switched systems with all or some stable subsystems earlier in the literature \cite{Narendra1994, Agrachev2012, abc, def}. However, to the best of our knowledge, this is the first instance where a blend of directed graphs and matrix commutators is employed to address stabilizability of the switched system \eqref{e:swsys} with all unstable subsystems under purely time-dependent switching signals that are not restricted to periodic constructions. Our stability conditions involve scalar inequalities, and are, therefore, numerically easier to verify compared to the existing matrix inequalities based conditions. We present a numerical experiment to demonstrate our results.
    	
\subsection{Paper organization}
\label{ss:pap_org}
    The remainder of this paper is organized as follows: we catalog the required preliminaries for our results in \S\ref{s:prelims}. Our main results are presented in \S\ref{s:mainres}. We also discuss various features of our results in this section. We present a numerical experiment in \S\ref{s:num_ex} and conclude in \S\ref{s:concln} with a brief discussion of future research directions. A proof of our main result is presented in \S\ref{s:all_proofs}.

    
\section{Preliminaries}
\label{s:prelims}
	Notice that if the family of systems \eqref{e:family} admits a Schur stable subsystem \(\ell\in\P\), then the set of stabilizing switching signals for \eqref{e:swsys} is necessarily non-empty. Indeed, a switching signal \(\sigma\) that obeys \(\sigma(t) = \ell\) for all \(t\), is stabilizing. However, we do not consider the presence of Schur stable subsystems. Given a set of unstable matrices \(A_{\ell}\), \(\ell\in\P\), it may be the case that there exists no switching signal \(\sigma\) that stabilizes \eqref{e:swsys}. In the sequel we will find sufficient conditions on the matrices \(A_{\ell}\), \(\ell\in\P\), such that there exists a \(\sigma\) under which the switched system \eqref{e:swsys} is GES. Prior to presenting our solution to Problem \ref{prob:mainprob}, we catalog a few preliminaries.

    We will assume that the family of systems \eqref{e:family} admits two unstable systems that form a Schur stable combination of the following nature:
   \begin{assump}
   \label{a:stable_combi}
   \rm{
        There exist \(i,j\in\P\) and \(p,q\in\N\) such that the matrix \(A_{N+1} := A_{i}^{p}A_{j}^{q}\) is Schur stable.
   }
   \end{assump}
   \begin{rem}
   \label{rem:periodic_vs_non-periodic}
   \rm{
        In view of Assumption \ref{a:stable_combi}, it is evident that a switching signal \(\sigma\) that obeys the following:
        \begin{align*}
            \sigma(0) &= j,\\
            \tau_{k+1}-\tau_{k} &=
            \begin{cases}
                q,\:\:\text{if}\:\sigma(\tau_{k})=j,\\
                p,\:\:\text{if}\:\sigma(\tau_{k}) = i,
            \end{cases}
            k = 0,1,2,\ldots,
        \end{align*}
        ensures stability of \eqref{e:swsys}. Indeed, we are dealing with stability of the difference equation
        \[
            x(\tilde{t}+1) = A_{N+1} x(\tilde{t}),\:\:x(0) = x_{0},
        \]
        where \(A_{N+1}\) is Schur stable and \(\tilde{t}:t = (p+q):1\). Clearly, if  Assumption \ref{a:stable_combi} holds, then we have a periodic switching signal that stabilizes \eqref{e:swsys}, and hence satisfaction of Assumption \ref{a:stable_combi} is a trivial solution to Problem  \ref{prob:mainprob}. However, this solution does not activate all subsystems unless \(\P = \{1,2\}\). In this paper we intend to transcend beyond periodic construction of stabilizing switching signals and allow the activation of all subsystems.
   }
   \end{rem}

   The following fact is known from the properties of Schur stable matrices:
   \begin{fact}
   \label{fact:m_defn}
   \rm{
        There exist \(m\in\N\) and \(\rho\in]0,1[\) such that the following condition holds:
        \begin{align}
        \label{e:m_ineq}
            \norm{A_{N+1}^{m}}\leq\rho.
        \end{align}
   }
   \end{fact}

   Let
   \begin{align}
   \label{e:M12_defn}
        M_{1} = \max_{\ell\in\P}\norm{A_{\ell}}\:\:\text{and}\:\:M_{2} = \norm{A_{N+1}}.
   \end{align}

   We will employ the following set of (matrix) commutators as the main apparatus for our analysis:
   \begin{align}
   \label{e:E_comm}
        E_{\ell,N+1} = A_{\ell}A_{N+1} - A_{N+1} A_{\ell},\:\:\ell\in\P.
   \end{align}
   Notice that \(E_{\ell, N+1}\) is the commutator between the matrix \(A_{\ell}\) and the matrix product \(A_{i}^{p}A_{j}^{q}\), \(\ell\in\P\), where \(i\), \(j\), \(p\) and \(q\) are as described in Assumption \ref{a:stable_combi}.

   We are now in a position to present our main result.
\section{Main results}
\label{s:mainres}
    We will solve Problem \ref{prob:mainprob} in two steps:
    \begin{itemize}[label = \(\circ\), leftmargin = *]
        \item First, we will present an algorithm to construct switching signals \(\sigma\), and
        \item second, we will find sufficient conditions on the subsystem matrices, \(A_{\ell}\), \(\ell\in\P\), such that the switched system \eqref{e:swsys} is GES under switching signals obtained from our algorithm.
    \end{itemize}

\subsection{Algorithmic construction of switching signals}
\label{ss:sw_algo}
    \begin{algorithm*}[htbp]
	\caption{Construction of switching signals}\label{algo:sw-sig_construc}
		\begin{algorithmic}[1]
			\renewcommand{\algorithmicrequire}{\textbf{Input:}}
			\renewcommand{\algorithmicensure}{\textbf{Output:}}
			
			\REQUIRE A family of systems \eqref{e:family} such that Assumptions \ref{a:all_unstable} and \ref{a:stable_combi} hold.
            \ENSURE A switching signal, \(\sigma\).
         		\STATE {\bf Step I}: Construct a directed graph \(G(V,E)\) as follows:\label{step:graph}
                \STATE Create a vertex set \(V = \P\cup\{N+1\}\).
                \STATE Create an edge set
                \begin{align*}
                    \hspace*{2cm}E = \{(\ell,\ell+1),\ell\in\P\setminus\{N\}\}\cup\{(\ell,N+1),\ell\in\P\}
                    \cup\{(N+1,\ell),\ell\in\P\}.
                \end{align*}
           \STATE {\bf Step II}: Construct a switching signal \(\sigma\) as an infinite walk on \(G(V,E)\) as follows:\label{step:sw-sig}
                \STATE Set \(k = 0\) and \(\tau_{k} = 0\).
                \STATE Pick a vertex \(v_{k}\in V\).
                \STATE Set \(v^{-} = v_{k}\). \label{step:repeat}
                \IF {\(v_{k}\in\P\)}
                    \STATE Set \(\sigma(\tau_{k}) = v_{k}\).\
                  \STATE Set \(k = k+1\) and \(\tau_{k} = \tau_{k-1}+1\).
               \ELSE
                   \STATE Set \(\sigma(\tau_{k}) = j\).\
                    \STATE Set \(k = k+1\) and \(\tau_{k} = \tau_{k-1}+q\).
                    \STATE Set \(\sigma(\tau_{k}) = i\).\
                    \STATE Set \(k = k+1\) and \(\tau_{k} = \tau_{k-1}+p\).
		\ENDIF
            \STATE Pick a vertex \(v_{k}\in V\) such that \((v^{-},v_{k})\in E\).
            \STATE Go to \ref{step:repeat}.
		\end{algorithmic}
	\end{algorithm*}

        Given a family of systems \eqref{e:family} such that Assumptions \ref{a:all_unstable} and \ref{a:stable_combi} hold, Algorithm \ref{algo:sw-sig_construc} designs a switching signal \(\sigma\) in two steps:
        \begin{itemize}[label = \(\circ\), leftmargin = *]
        		\item In \ref{step:graph} we construct a directed graph \(G(V,E)\). It has \(N+1\) vertices in which the vertices \(1,2,\ldots,N\) correspond to the given \(N\) subsystems and the vertex \(N+1\) corresponds to the Schur stable combination formed by subsystems \(A_{i}\) and \(A_{j}\), \(i,j\in\P\). Each vertex (subsystem) \(\ell\in V\setminus\{N, N+1\}\) has an outgoing edge to the vertex (subsystem) \(\ell+1\), and each vertex \(\ell\in\P\) has an outgoing edge to the vertex \(N+1\) and an incoming edge from the vertex \(N+1\). These edges correspond to switches between the subsystems that we allow to be caused by our switching signals.
		\item In \ref{step:sw-sig} we construct a switching signal \(\sigma\) as an infinite walk on the directed graph \(G(V,E)\).\footnote{Recall that a \emph{walk} on a directed graph is an alternating sequence of vertices and edges \(W = v_{0}, e_{1}, v_{1}, e_{2}, v_{2}, \ldots, v_{n-1}, e_{n}, v_{n}\), where \(v_{\ell}\in V\), \(e_{\ell} = (v_{\ell-1},v_{\ell})\in E\), \(0 < \ell \leq n\). The length of a walk is its number of edges, counting repetitions, e.g., the length of \(W\) above is \(n\). By the term \emph{infinite walk} we mean a walk of infinite length, i.e., it has infinitely many edges.}  For the vertices \(1,2,\ldots,N\), the corresponding subsystem is activated for \(1\) unit of time, and for the vertex \(N+1\), subsystem \(j\) is activated for \(q\) unit(s) of time followed by activation of subsystem \(i\) for \(p\) unit(s) of time.
        \end{itemize}
        By the structure of \(G(V,E)\), a \(\sigma\) obtained from Algorithm \ref{algo:sw-sig_construc} is not restricted to periodic constructions.

    \begin{rem}
    \label{rem:graph_construc}
    \rm{
    	In Algorithm \ref{algo:sw-sig_construc} the activation of the vertex \(N+1\) corresponds to dwelling on subsystem \(j\) for \(q\) units of time followed by dwelling on subsystem \(i\) for \(p\) units of time. Notice that by construction of \(G(V,E)\), two consecutive instances of such activations are at most \(N\) units of time apart. In particular, we allow the switches from subsystem \(\ell\) to subsystem \(\ell+1\), \(\ell\in\P\setminus\{N\}\) and not vice-versa, to exclude the possibility that the Schur stable combination \(A_{N+1}\) is never caused by \(\sigma\). In fact, our analysis works for one-directional switching between the subsystems \(\ell\in\P\) chosen in any order of the indices.
    }
    \end{rem}

\subsection{Sufficient conditions for stability of \eqref{e:swsys}}
\label{ss:stab_condns}
	It is evident that given a family of systems \eqref{e:family}, whether a switching signal \(\sigma\) obtained from Algorithm \ref{algo:sw-sig_construc} stabilizes \eqref{e:swsys} or not, depends on the properties of the subsystem matrices \(A_{\ell}\), \(\ell\in\P\). The following theorem identifies sufficient conditions on \(A_{\ell}\), \(\ell\in\P\), such that a \(\sigma\) obtained from our algorithm ensures GES of \eqref{e:swsys}:
	\begin{thm}
	\label{t:mainres}
		Consider a family of systems \eqref{e:family}. Suppose that Assumptions \ref{a:all_unstable} and \ref{a:stable_combi} hold. Let \(\lambda\) be an arbitrary positive number satisfying
		\begin{align}
		\label{e:maincondn1}
			\rho e^{\lambda m(p+q)} < 1.
		\end{align}
		Suppose that there exists a scalar \(\varepsilon\) small enough such that the following conditions hold:
		\begin{align}
		\label{e:maincondn2}
			\norm{E_{\ell,N+1}}\leq\varepsilon\:\:\text{for all}\:\:\ell\in\P,
		\end{align}
		and
		\begin{align}
		\label{e:maincondn3}
			\rho e^{\lambda m(p+q)} + N\frac{m(m+1)}{2}M_{1}^{mN-1}M_{2}^{m-1}\varepsilon e^{\lambda{m(p+q)}+\lambda mN} \leq 1.
		\end{align}
		Then the switched system \eqref{e:swsys} is GES under a switching signal obtained from Algorithm \ref{algo:sw-sig_construc}.
	\end{thm}
	
	Theorem \ref{t:mainres} is our solution to Problem \ref{prob:mainprob}; a proof of it is presented in \S\ref{s:all_proofs}. Given a family of systems \eqref{e:family} such that Assumptions \ref{a:all_unstable} and \ref{a:stable_combi} hold, in Algorithm \ref{algo:sw-sig_construc} we construct a switching signal \(\sigma\). Theorem \ref{t:mainres} provides sufficient conditions on the subsystem matrices \(A_{\ell}\), \(\ell\in\P\), under which such a \(\sigma\) ensures GES of the switched system \eqref{e:swsys}. Since \(\rho < 1\), it is always possible to find a \(\lambda > 0\) (could be small) such that condition \eqref{e:maincondn1} holds. If in addition, the Euclidean norms of (matrix) commutators of \(A_{\ell}\) and \(A_{N+1}\), \(\ell\in\P\), are bounded above by a scalar \(\varepsilon\) small enough such that condition \eqref{e:maincondn3} holds, then \eqref{e:swsys} is GES under a \(\sigma\) obtained from Algorithm \ref{algo:sw-sig_construc}. In the simplest case when the matrices \(A_{\ell}\) and \(A_{N+1}\) commute for all \(\ell\in\P\), condition \eqref{e:maincondn3} reduces to condition \eqref{e:maincondn1}. Theorem \ref{t:mainres} accommodates sets of matrices \(A_{\ell}\), \(\ell\in\P\), for which \(A_{\ell}\) and \(A_{N+1}\), \(\ell\in\P\), do not necessarily commute, but are ``close'' to sets of matrices for which they commute. This feature associates an inherent robustness with our stabilizability conditions. Indeed, if we are relying on approximate models of \(A_{\ell}\), \(\ell\in\P\), or the elements of \(A_{\ell}\), \(\ell\in\P\), are prone to evolve over time, then GES of \eqref{e:swsys} holds under a \(\sigma\) obtained from Algorithm \ref{algo:sw-sig_construc} as long as the commutators of \(A_{\ell}\) and \(A_{N+1}\), \(\ell\in\P\), in their Euclidean norm, are bounded above by a small scalar \(\varepsilon\) such that condition \eqref{e:maincondn3} holds.

    \begin{rem}
    \label{rem:stable_combi}
    \rm{
    	The motivation behind Assumption \ref{a:stable_combi} is the use of purely time-dependent stabilizing switching signals. In general, if the subsystem matrices \(A_{\ell}\), \(\ell\in\P\), do not form any type of Schur stable combination, then information about the system state \(x(t)\), \(t\in\N_{0}\), is required to guarantee global stability of \eqref{e:swsys}. However, for complex systems the access to full state information may not be available at every time instant.
    }
    \end{rem}
    
    \begin{rem}
    \label{rem:state-vs-time}
    \rm{
    	A vast body of the literature relies on state-dependent switching signals for stabilizability of the switched system \eqref{e:swsys}. In contrast, the recent works \cite{Fiacchini2016,Kun_arxiv2020} deal with purely time-dependent switching signals. In \cite{Fiacchini2016} the authors propose (among others) purely time-dependent (in particular, periodic) switching signals that ensure stability of \eqref{e:swsys}. They show that if a family of systems \eqref{e:family} admits a stabilizing switching signal, then it is not necessary that it admits a stabilizing periodic switching signal. They also propose necessary and sufficient conditions for the existence of stabilizing periodic switching signals based on LMIs. Our results differ from \cite{Fiacchini2016} in the following two aspects: 
	\begin{itemize}[label = \(\circ\), leftmargin = *]
		\item first, our switching signals are not restricted to periodic construction, and
		\item second, our stability conditions involve scalar inequalities.
	\end{itemize}
	In \cite{Kun_arxiv2020} the author considers stabilizability of \eqref{e:swsys} under restrictions on admissible switches between the subsystems and admissible dwell times on the subsystems. The proposed stability conditions involve matrix commutator conditions and graph theoretic arguments. Our work differs from \cite{Kun_arxiv2020} in the following two aspects:
	\begin{itemize}[label = \(\circ\), leftmargin = *]
		\item first, our switching signals are unrestricted, thereby allowing us to choose the subsystems to be activated and the dwell times on them as per our requirements, and
		\item second, we employ matrix commutators between a Schur stable combination of two subsystems matrices and each individual subsystem matrix, while in \cite{Kun_arxiv2020} matrix commutators of products of subsystems matrices individually and with others are employed.
	\end{itemize}
	The choices of matrix commutators under consideration in both \cite{Kun_arxiv2020} and this paper are, however, not unique. Commutation relations between the subsystems matrices or certain products of these matrices have also been employed to study stability of the switched system \eqref{e:swsys} with some or all stable subsystems and restricted or unrestricted switching, see e.g.,  \cite{Narendra1994,Agrachev2012,abc,def, stab_cycle}. The crux of the (matrix) commutation relations based stability analysis for switched systems lies in splitting matrix products into sums and applying combinatorial arguments on them, see \cite{Agrachev2012}, where this analysis technique was introduced. In this paper we treat stabilizability of \eqref{e:swsys} under purely time-dependent switching signals constructed as infinite walks on a directed graph. Our switching signals are not restricted to periodic construction, and the proposed conditions on \(A_{\ell}\), \(\ell\in\P\), for stability are only sufficient. The activation of the subsystems \(\ell\in\P\) for \(1\) unit of time, and the subsystems \(i\) and \(j\) for \(p\) and \(q\) units of time, respectively, to form a Schur stable combination, facilitates us to employ the matrix commutators between the matrices \(A_{\ell}\), \(\ell\in\P\), and the matrix \(A_{N+1}\). See \S\ref{s:all_proofs} for details of the procedure of splitting a matrix product into sums and applying combinatorial arguments on them. One may employ different choices of commutators to split a matrix product into sums, thereby leading to different sets of sufficient solutions to Problem \ref{prob:mainprob}.
	}
    \end{rem}
    
%
    \begin{rem}
    \label{rem:matrix-vs-scalar_inequality}
    \rm{
    	Prior results on stabilizability of \eqref{e:swsys} are primarily based on satisfaction of certain sets of matrix inequalities. For instance, stability of \eqref{e:swsys} under the so-called min switching signal \cite{Heemels2017} is guaranteed if the subsystem matrices \(A_{\ell}\), \(\ell\in\P\), satisfy the Lyapunov-Metzler inequalities \cite{Geromel2006} or the S-procedure characterization \cite{Heemels2017}. These are BMIs, and are numerically difficult to verify. The necessary and sufficient condition for the existence of a stabilizing periodic switching signal proposed in \cite{Fiacchini2016} involves satisfaction of LMI-based conditions. In contrast, our stability conditions rely on scalar inequalities involving upper bounds on the Euclidean norms of a set of commutators of certain combinations of the subsystem matrices, a set of scalars obtained from the properties of these matrices, and the total number of subsystems. The satisfaction of Assumption \ref{a:stable_combi} can be checked by using the LMI based condition presented in \cite{Fiacchini2016}.
    }
    \end{rem}

\section{A numerical experiment}
\label{s:num_ex}
	We design the following two routines in MATLAB R2020b on an OS-X platform:
	\begin{enumerate}[leftmargin = *]
		\item The first routine checks our stability conditions.
		\begin{itemize}[label = \(\circ\), leftmargin = *]
			\item We consider a family of systems \eqref{e:family} with \(N = 10\). We generate the matrices \(A_{\ell}\in\R^{2\times 2}\), \(\ell\in\P\), by picking elements from the interval \([-1,1]\) uniformly at random. It is ensured that all the matrices are unstable.
			\item We check the conditions of Theorem \ref{t:mainres} and obtain 
			\begin{align*}
    	&\rho e^{\lambda m} + N\frac{m(m+1)}{2}M_1^{mN-1}M_2^{m-1}\varepsilon e^{\lambda m(p+q)+\lambda m N}\\  
	&= 0.9235 < 1.
    \end{align*}
		\end{itemize}
		\item The second routine generates stabilizing switching signals.
			\begin{itemize}[label = \(\circ\), leftmargin = *]
				\item We employ Algorithm \ref{algo:sw-sig_construc} to generate a switching signal \(\sigma\).
				\item We pick \(100\) different initial conditions from the interval \([-1,1]^{2}\) uniformly at random and study the solution of the switched system \eqref{e:swsys} under the above \(\sigma\). GES of \eqref{e:swsys} is observed. We plot \((\sigma(t))_{t\in\N_0}\) and its corresponding \((\norm{x(t)})_{t\in\N_0}\) in Figures \ref{fig:plot2} and \ref{fig:plot1}, respectively.
			\end{itemize}
	\end{enumerate} 
	\begin{figure}[htbp]
	\centering
		\includegraphics[scale=0.5]{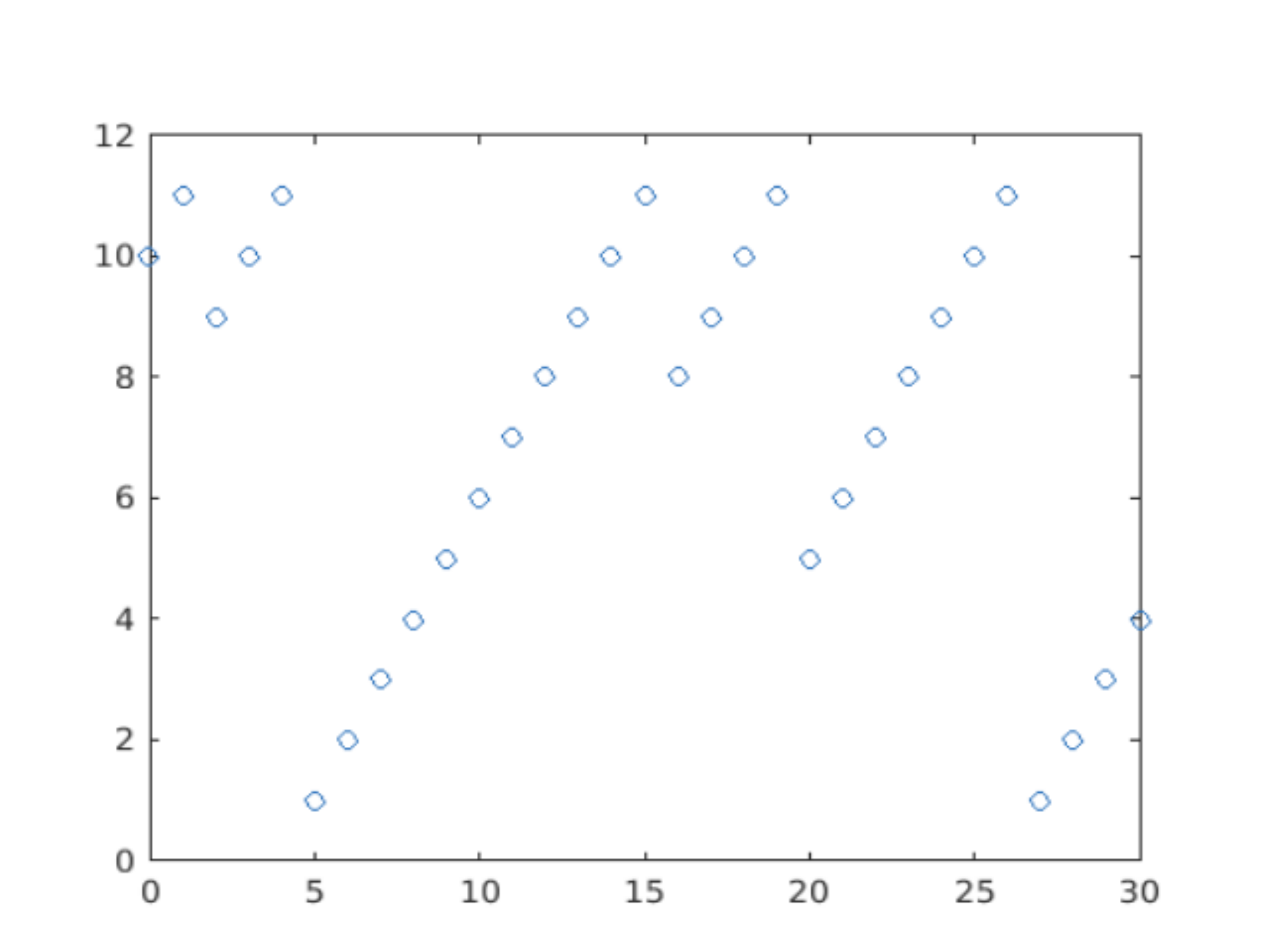}
		\caption{\((\sigma(t))_{t\in\N_0}\) obtained from Algorithm \ref{algo:sw-sig_construc}}\label{fig:plot2}
	\end{figure}

	\begin{figure}[htbp]
	\centering
		\includegraphics[scale=0.5]{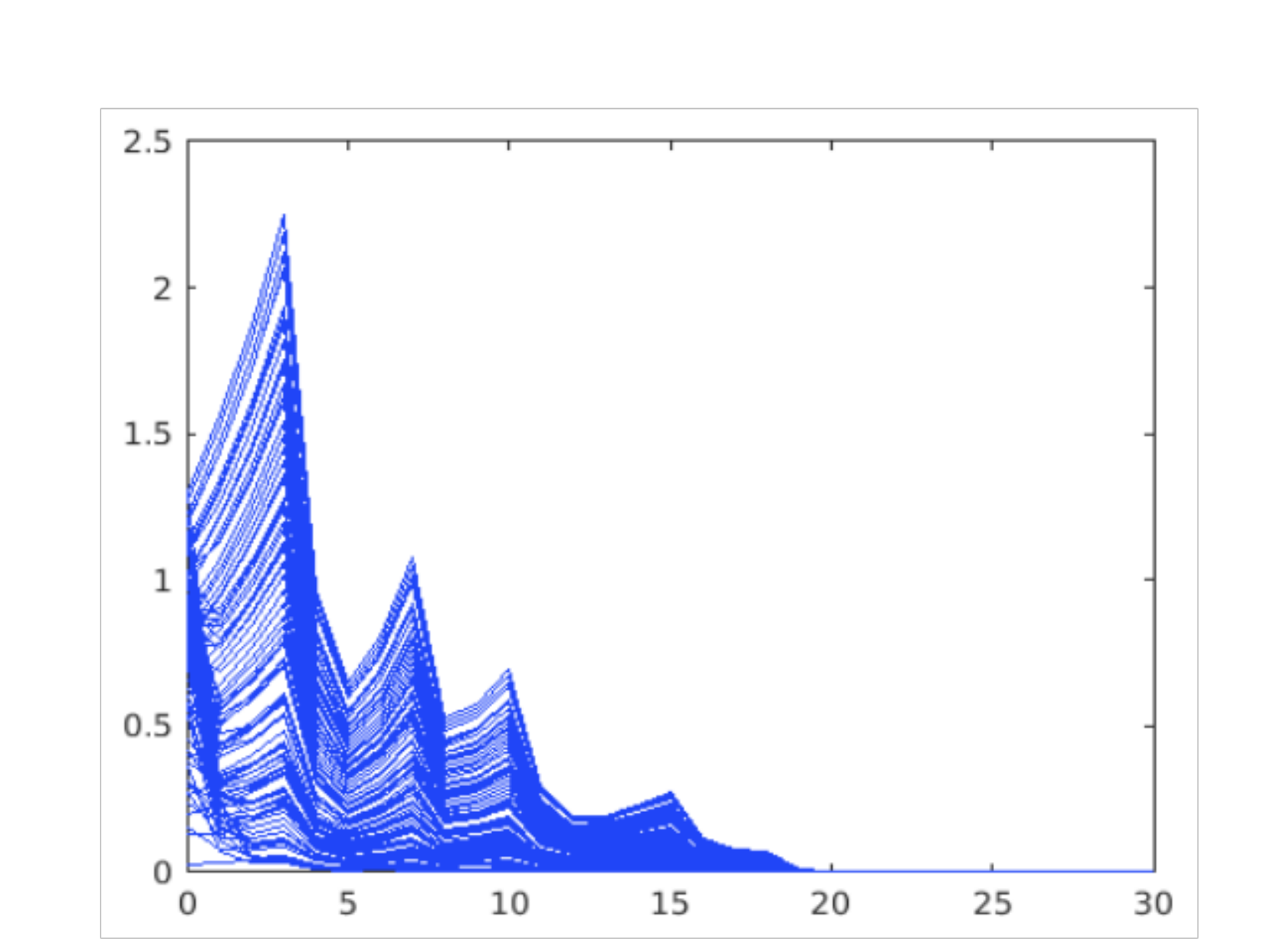}
		\caption{\((\norm{x(t)})_{t\in\N_0}\) under \(\sigma\) demonstrated in Figure \ref{fig:plot2}}\label{fig:plot1}
	\end{figure}
\section{Concluding remarks}
\label{s:concln}
	In this paper we studied stabilizability of switched linear systems whose all subsystems are unstable. Our stabilizing switching signals are purely time-dependent, are allowed to activate all subsystems and are not restricted to periodic construction. The stability conditions proposed in this paper rely on the existence of a Schur stable combination formed by two subsystems and involve scalar inequalities. We identify the following two directions for future work: first, determining the relations between Theorem \ref{t:mainres} and the existing sufficient conditions (\'{a} la \cite{Fiacchini2016}) and second, the extension of our techniques to stabilizability of switched nonlinear systems. The findings on these topics will be reported elsewhere.
\section{Proof of our result}
\label{s:all_proofs}
    Fix a switching signal \(\sigma\) obtained from Algorithm \ref{algo:sw-sig_construc}. Let \(M\) be the corresponding word (matrix product) defined as
    \[
        M = \cdots A_{\sigma(2)}A_{\sigma(1)}A_{\sigma(0)}.
    \]
    The condition \eqref{e:gues1} for GES of \eqref{e:swsys} under \(\sigma\) can be written equivalently as \cite[\S2]{Agrachev2012}: there exist positive numbers \(c\) and \(\lambda\) such that
    \begin{align}
    \label{e:gues2}
        \norm{M}\leq ce^{-\lambda\abs{M}}\:\:\text{for all}\:\:\abs{M},
    \end{align}
    where \(\abs{M}\) denotes the length of the product \(M\), i.e., the number of matrices that appear in \(M\), counting repetitions.

    It, therefore, suffices to show that if the conditions of Theorem \ref{t:mainres} hold, then condition \eqref{e:gues2} is satisfied under a \(\sigma\) obtained from Algorithm \ref{algo:sw-sig_construc}.

    \begin{proof}[Proof of Theorem \ref{t:mainres}]
        Let \(\sigma\) be a switching signal obtained from Algorithm \ref{algo:sw-sig_construc}, and \(M\) be the corresponding matrix product. We will show that if the subsystem matrices \(A_{\ell}\), \(\ell\in\P\), satisfy the conditions of Theorem \ref{t:mainres}, then the condition \eqref{e:gues2} is true. We will employ mathematical induction on \(\abs{M}\) for this purpose.

        {\it A. Induction basis}: Pick \(c\) large enough so that \eqref{e:gues2} holds for \(M\) satisfying \(\abs{M}\leq m(p+q)+mN\).

        {\it B. Induction hypothesis}: Let \(\abs{M}\geq m(p+q)+mN+1\) and assume that \eqref{e:gues2} holds for length less than \(\abs{M}\).

        {\it C. Induction step}: Let \(M = LR\), where \(\abs{R} = m(p+q)+mN\). We observe that \(R\) contains at least \(m\) instances of the matrix \(A_{N+1}\). Indeed, this property of \(R\) follows from the construction of \(\sigma\) as an infinite walk on \(G(V,E)\).

        Let us rewrite \(R\) as
        \begin{align}
        \label{e:pf1_step1}
            R = R_{1}A_{N+1}^{m} + R_{2},
        \end{align}
        where
        \begin{align*}
            \abs{R_{1}} &= m(p+q) +mN - \abs{A_{N+1}^{m}} = mN.
        \end{align*}

        The following two possibilities exist:\\
        \begin{enumerate}[label = Case \Roman*., leftmargin = *]
            \item If \(\sigma(0)\in\P\), then \(R_{2}\) contains at most \(\displaystyle{N\frac{m(m+1)}{2}}\) terms of length \((m-1)(p+q)+mN\) with \(m-1\) \(A_{N+1}\), \(mN-1\) \(A_{\ell}\), \(\ell\in\P\), and \(1\) \(E_{\ell,N+1}\), \(\ell\in\P\) generated by exchanging the \(k\)-th instance (reading from the right) of \(A_{N+1}\) in \(R\) with at most \(kN\) instances of \(A_{\ell}\), \(\ell\in\P\), towards achieving the structure of \(R\) given in \eqref{e:pf1_step1}, \(k=1,2,\ldots,m\).
                             \item If \(\sigma(0) = N+1\), then \(R_{2}\) contains at most \(\displaystyle{N\frac{m(m-1)}{2}}\) terms of length \((m-1)(p+q)+mN\) with \(m-1\) \(A_{N+1}\), \(mN-1\) \(A_{\ell}\), \(\ell\in\P\), and \(1\) \(E_{\ell.N+1}\), \(\ell\in\P\) generated by exchanging the \(k\)-th instance (reading from the right) of \(A_{N+1}\) in \(R\) with \((k-1)N\) instances of \(A_{\ell}\), \(\ell\in\P\), towards achieving the structure of \(R\) in \eqref{e:pf1_step1}, \(k=1,2,\ldots,m\).\\

        \end{enumerate}

        Now, applying the sub-multiplicativity and sub-additivity properties of the induced Euclidean norm, we obtain
        \begin{align}
        \label{e:pf1_step2}
            &\norm{M} = \norm{LR} = \norm{L(R_{1}A_{N+1}^{m}+R_{2})}\nonumber\\
            &\leq \norm{LR_{1}}\norm{A_{N+1}^{m}}+\norm{L}\norm{R_{2}}\nonumber\\
            &\leq\rho c e^{-\lambda(\abs{M}-m(p+q))}+ce^{-\lambda(\abs{M}-m(p+q)-mN)}\norm{R_{2}},
        \end{align}
        where the upper bounds on \(\norm{LR_{1}}\) and \(\norm{L}\) are obtained by using the relations \(\abs{M} = \abs{LR_{1}} + \abs{A_{N+1}^{m}}\) and \(\abs{M} = \abs{L}+\abs{R}\), respectively.

         \begin{enumerate}[label = Case \Roman*., leftmargin = *]
            \item If \(\sigma(0)\in\P\), then
            \begin{align}
            \label{e:pf1_step3}
                \norm{R_{2}} \leq N\frac{m(m+1)}{2}M_{1}^{mN-1}M_{2}^{m-1}\varepsilon.
            \end{align}
            From \eqref{e:pf1_step2} and \eqref{e:pf1_step3}, we have that
            \begin{align}
            \label{e:pf1_step4}
                \norm{M}&\leq\rho c e^{-\lambda(\abs{M}-m(p+q))}+ce^{-\lambda(\abs{M}-m(p+q)-mN)}
                \times N\frac{m(m+1)}{2}M_{1}^{mN-1}M_{2}^{m-1}\nonumber\\
                &=ce^{-\lambda\abs{M}}\Biggl(\rho e^{\lambda m(p+q)}+N\frac{m(m+1)}{2}M_{1}^{mN-1}
                \times M_{2}^{m-1}\varepsilon e^{\lambda m(p+q) + \lambda m N}\Biggr).
            \end{align}
            Applying \eqref{e:maincondn3} to \eqref{e:pf1_step4}, we have that \eqref{e:gues2} holds.
            \item If \(\sigma(0) = N+1\), then
            \begin{align}
            \label{e:pf1_step5}
                \norm{R_{2}} &\leq N\frac{m(m-1)}{2}M_{1}^{mN-1}M_{2}^{m-1}\varepsilon\nonumber\\
                &\leq N\frac{m(m+1)}{2}M_{1}^{mN-1}M_{2}^{m-1}\varepsilon.
            \end{align}
            The condition \eqref{e:gues2} follows under the same set of arguments as for Case I.
         \end{enumerate}

         This completes our proof of Theorem \ref{t:mainres}.
    \end{proof}

\end{document}